\documentclass[11pt,letterpaper,reqno]{amsart}

\usepackage{mathpazo}
\usepackage{color}
\usepackage{array}
\usepackage{amsmath}
\usepackage{amssymb}
\usepackage{amsthm}
\newtheorem{theorem}{Theorem}[section]
\newtheorem{lemma}[theorem]{Lemma}

\newtheorem{corollary}[theorem]{Corollary}

\newtheorem{remark}[theorem]{Remark}

\usepackage[alphabetic,nobysame]{amsrefs}

\usepackage[margin=1in]{geometry}
\usepackage[colorlinks,citecolor=blue]{hyperref}

\newcommand{\eq}[1]{\hyperref[eq:#1]{Equation~(\ref*{eq:#1})}}
\renewcommand{\sec}[1]{\hyperref[sec:#1]{Section~\ref*{sec:#1}}}
\newcommand{\thm}[1]{\hyperref[thm:#1]{Theorem~\ref*{thm:#1}}}
\newcommand{\lem}[1]{\hyperref[lem:#1]{Lemma~\ref*{lem:#1}}}
\newcommand{\cor}[1]{\hyperref[cor:#1]{Corollary~\ref*{cor:#1}}}
\newcommand{\itm}[1]{\hyperref[itm:#1]{\ref*{itm:#1}}}
\newcommand{\app}[1]{\hyperref[app:#1]{Appendix~\ref*{app:#1}}}

\newcommand{\abs}[1]{\left| #1 \right|}
\newcommand{\prn}[1]{\left( #1 \right)}
\newcommand{\brc}[1]{\left\{ #1 \right\}}

\newcommand{\bra}[1]{\left\langle #1 \right\vert}
\newcommand{\ket}[1]{\left\vert #1 \right\rangle}
\newcommand{\bkt}[1]{\left\langle #1 \right\rangle}

\newcommand{\F}[0]{\mathbb{F}}
\newcommand{\Fq}[0]{\mathbb{F}_q}
\newcommand{\I}[0]{\mathbb{I}}
\newcommand{\J}[0]{\mathbb{J}}
\newcommand{\K}[0]{\mathbb{K}}
\newcommand{\R}[0]{\mathbb{R}}

\newcommand{\Z}[0]{\mathbb{Z}}
\newcommand{\C}[0]{\mathbb{C}}
\newcommand{\N}[0]{\mathbb{N}}

\newcommand{\dd}[0]{\mathrm{d}}
\newcommand{\tr}[0]{\mathrm{Tr}}
\newcommand{\FF}[0]{\mathcal{F}}

\newcommand{\supp}[0]{\mathrm{supp}}
\newcommand{\nul}[1]{{}}
\newcommand{\conj}[1]{\overline{{#1}}}
\newcommand{\vc}[0]{\mathbf{c}}
\newcommand{\vf}[0]{\mathbf{f}}
\newcommand{\vu}[0]{\mathbf{u}}
\newcommand{\vx}[0]{\mathbf{x}}
\newcommand{\vy}[0]{\mathbf{y}}
\newcommand{\vz}[0]{\mathbf{z}}
\newcommand{\spn}{\mathrm{span}}
\renewcommand{\Re}{\mathrm{Re}}
\renewcommand{\Im}{\mathrm{Im}}

\begin{document}

\title{Quantum algorithm for multivariate polynomial interpolation}
\author[Jianxin Chen]{Jianxin Chen$^{\ddag}$}
\author[Andrew M.\ Childs]{Andrew M.\ Childs$^{*,\dagger,\ddag}$}
\author[Shih-Han Hung]{Shih-Han Hung$^{*,\ddag}$}

\renewcommand*{\thefootnote}{\fnsymbol{footnote}}
\footnotetext[1]{\textsc{Department of Computer Science, University of Maryland}}
\footnotetext[2]{\textsc{Institute for Advanced Computer Studies, University of Maryland}}
\footnotetext[3]{\textsc{Joint Center for Quantum Information and Computer Science, University of Maryland}}

\renewcommand*{\thefootnote}{\arabic{footnote}}


\begin{abstract}
How many quantum queries are required to determine the coefficients of a degree-$d$ polynomial in $n$ variables?  We present and analyze quantum algorithms for this multivariate polynomial interpolation problem over the fields $\Fq$, $\R$, and $\C$.  We show that $k_{\C}$ and $2k_{\C}$ queries suffice to achieve probability $1$ for $\C$ and $\R$, respectively, where $k_{\C}=\smash{\lceil\frac{1}{n+1}{n+d\choose d}\rceil}$ except for $d=2$ and four other special cases. For $\Fq$, we show that $\smash{\lceil\frac{d}{n+d}{n+d\choose d}\rceil}$ queries suffice to achieve probability approaching $1$ for large field order $q$. The classical query complexity of this problem is $\smash{n+d\choose d}$, so our result provides a speedup by a factor of $n+1$, $\frac{n+1}{2}$, and $\frac{n+d}{d}$ for $\C$, $\R$, and $\Fq$, respectively. Thus we find a much larger gap between classical and quantum algorithms than the univariate case, where the speedup is by a factor of $2$. For the case of $\Fq$, we conjecture that $2k_{\C}$ queries also suffice to achieve probability approaching $1$ for large field order $q$, although we leave this as an open problem.
\end{abstract}

\maketitle


\section{Introduction}\label{sec:introduction}

Let $f(x_1,\ldots,x_n)\in \K[x_1,\ldots,x_n]$ be a polynomial of degree $d$. Suppose $d$ is known and we are given a black box that computes $f$ on any desired input. The polynomial interpolation problem is to determine all the coefficients of the polynomial by querying the black box.

Classically, a multivariate polynomial can be interpolated by constructing a system of linear equations. Invertibility of the Vandermonde matrix implies that $\smash{\binom{n+d}{d}}$ queries are necessary and sufficient to determine all the coefficients.
(Note that one must choose the input values carefully to construct a full-rank Vandermonde matrix for $n>1$ \cite{GS00}.)

Recent work has established tight bounds on the quantum query complexity of interpolating univariate polynomials over a finite field $\Fq$.
In particular, \cite{CvDHS16} developed an optimal quantum algorithm that
makes $\frac{d+1}{2}$ queries to succeed with bounded error and
one more query to achieve success probability $1-O(1/q)$. They also showed that the success probability of the algorithm is optimal among all algorithms making the same number of queries.  Previous work \cite{KK11,MP11} shows that no quantum algorithm can succeed
with bounded error using fewer queries, so the optimal success probability exhibits a sharp transition as the number of queries is increased.

For multivariate polynomials, \cite{CvDHS16} conjectured that a straightforward analog of the univariate algorithm solves the interpolation problem with probability $1-o(1)$ using  $\smash{\lfloor\frac{1}{n+1}\binom{n+d}{d}\rfloor+1}$ queries.  However, while that conjecture is natural, the analysis of the algorithm appeared to require solving a difficult problem in algebraic geometry and was left open.  (In addition, Montanaro considered the quantum query complexity of interpolating a multilinear polynomial \cite{Mon12}, but this is quite different from the general multivariate case.)

To the best of our knowledge, all previous work on quantum algorithms for
polynomial interpolation has focused on finite fields. Cryptographic applications of interpolation typically use finite fields, and the multivariate case could lead to new applications in that domain.  However, polynomial interpolation over infinite fields is also a natural problem, especially considering the ubiquity of real- and complex-valued polynomials in numerical analysis.

In this paper, we propose an approach to the quantum query complexity of polynomial interpolation in the continuum limit. To obtain a well-defined initial state, the algorithm prepares a superposition over a bounded working region. The bounded region limits the precision that can be achieved due to the uncertainty principle, but the algorithm can be made arbitrarily precise by taking an arbitrarily large region. Using this strategy, we present a quantum algorithm for multivariate polynomial interpolation over the real and complex numbers. To simplify the analysis, we allow the algorithm to work with arbitrarily precise inputs and outputs over $\R$ or $\C$; in practice, sufficiently fine discretization of the space could achieve similar performance. We also consider multivariate polynomial interpolation over finite fields, where our algorithm can be viewed as a generalization of the univariate polynomial interpolation algorithm proposed in \cite{CvDHS16}.

To analyze the success probability of our approach, we relate it to the tensor rank problem. The rank of a given tensor, which is the smallest integer $k$ such that the tensor can be decomposed as linear combination of $k$ simple tensors (i.e., those that can be written as tensor products), was first introduced nearly a century ago. A half century later, with the advent of principal component analysis on multidimensional arrays, the study of tensor rank attracted further attention. However, it has recently been shown that most tensor problems, including tensor rank, are NP-hard \cite{H90,HL13,Shi16}, and restricting these problems to symmetric tensors does not seem to alleviate their NP-hardness \cite{HL13,Shi16}. More specifically, tensor rank is NP-hard over any field extension of $\mathbb{Q}$ and NP-complete over a finite field $\mathbb{F}_q$.

Fortunately, analyzing the success probability of multivariate polynomial interpolation does not require exactly computing the rank of a symmetric tensor. The number of queries needed to achieve success probability $1$ can be translated to the smallest integer $k$ such that \emph{almost every} symmetric tensor can be decomposed as a linear combination of no more than $k$ simple tensors. In turn, this quantity can be related to properties of certain secant varieties, which lets us take advantage of recent progress in algebraic geometry \cite{BBO15,BT15}.

The success probability of our algorithm behaves differently as a function of
the number of queries for the three fields we consider.
Specifically, by introducing
\begin{align}
k_{\C}(n,d):=
\begin{cases}
n+1 & \textrm{$d=2, n\geq 2$;}\\
\lceil\frac{1}{n+1}{\binom{n+d}{d}}\rceil+1 & \textrm{$(n,d)=(4,3),(2,4),(3,4),(4,4)$;}\\
\lceil\frac{1}{n+1}{\binom{n+d}{d}}\rceil  & \textrm{otherwise,}
\end{cases}
\end{align}
we have the following upper bounds on the query complexity:

\begin{theorem}\label{thm:main}
For positive integers $d$ and $n$,
there exists a quantum algorithm for interpolating an $n$-variate polynomial of
degree $d$ over the field $\K$ using at most
\begin{enumerate}
\item $\frac{d}{n+d}\binom{n+d}{d}$ queries for $\K=\Fq$, succeeding with probability $1-O(1/q)$;
\item $2k_\C$ queries for $\K=\R$, succeeding with probability $1$;
\item $k_\C$ queries for $\K=\C$, succeeding with probability $1$.
\end{enumerate}
\end{theorem}

Note that these upper bounds can be improved using known results \cite{CvDHS16} for univariate polynomial interpolation (see the final remark in \sec{performance-finite-field}).

The remainder of the paper is organized as follows. After introducing the notation and mathematical background in \sec{preliminaries}, we describe the query model in \sec{query-model} and present our algorithms in \sec{algorithm}. We then analyze the algorithm to establish our query complexity upper bounds in \sec{perf}. In \sec{opt}, we show that our proposed algorithm is optimal for finite fields. Finally, in \sec{open} we conclude by mentioning some open questions.

\section{Preliminaries and notations}\label{sec:preliminaries}

In \sec{basic}, we introduce some notation that is used in the paper, especially when describing the algorithm in \sec{algorithm}. \sec{alg-geometry} reviews basic definitions and concepts from algebraic geometry that arise in our performance analysis in \sec{perf}.

\subsection{Notation and definitions}\label{sec:basic}

Let $f\in \K[x_1,\ldots,x_n]$ be a polynomial of degree at most $d$ over the field
$\K$. We let $x^j:=\prod_{i=1}^n x_i^{j_i}$ for $j\in\J$, where
$\J:=\brc{j\in\N^n: j_1+\cdots+j_n\leq d}$ is the set of allowed exponents with
size $J:=\binom{n+d}{d}$. Thus $x^j$ is a monomial in $x_1,\ldots,x_n$ of degree $j_1+j_2+\cdots+j_n$. With this notation, we write $f(x)=\sum_{j \in \J}c_j x^j$ for some coefficients $\{c_j \in \K : j \in \J\}$.

Access to the function $f$ is given by a black box
that performs $\ket{x,y}\mapsto\ket{x,y+f(x)}$ for all $x \in \K^n$ and $y\in \K$. We will
compute the coefficients of $f$ by performing \emph{phase queries}, which are
obtained by phase kickback over $\K$, as detailed in \sec{query-model}.

For $k$-dimensional vectors $x,y\in
\K^k$, we consider the inner product $\cdot\colon \K^k\times \K^k\to \K$ defined by $x\cdot y =
\sum_{i=1}^k \conj{x}_iy_i$, where $\conj{x}$ is the complex conjugate of $x$
(where we let $\conj{x} = x$ for $x \in \F_q$).
We denote the indicator function for
a set $A\subseteq \R^n$ by $\I_A(z)$, which is $1$ if $z\in A$ and $0$ if $z \notin A$.
We denote a ball of radius $r\in\R^+$ centered at $0$ by $B(r)$.

A \emph{lattice} $\Lambda$ is a discrete additive subgroup of $\R^n$ for positive integer
$n$ generated by $e_1,\ldots,e_n\in\R^n$.
For every element $x\in\Lambda$, we have $x=\sum_{i=1}^n \gamma_i e_i$ for some $\gamma_i\in\Z$
for $i\in\brc{1,\ldots,n}$.
A \emph{fundamental domain} $T$ of $\Lambda$ centered at zero is a subset of $\R^n$ such that
$
T=\brc{\sum_{i=1}^na_ie_i: a_i\in\left[-\frac{1}{2},\frac{1}{2}\right)}
$.
The \emph{dual lattice} of $\Lambda$, denoted by $\widetilde{\Lambda}$, is an
additive subgroup of $\R^n$ generated by $f_1,\ldots,f_n\in\R^n$ satsifying $e_i\cdot f_j=\delta_{ij}$ for $i,j\in\brc{1,\ldots,n}$.

The standard basis over the real numbers is the set
$\brc{\ket{x}: x\in\R^n}$ for positive integer $n$.
The amplitude of a state $\ket{\psi}$ in the standard basis is denoted by
$\psi(x)$ or $\bkt{x|\psi}$.
The standard basis vectors over real numbers are orthonormal in the sense of the Dirac delta function, i.e., $\bkt{x'|x}=\delta^{(n)}(x-x')$ for $x,x'\in\R^n$.

\subsection{Algebraic geometry concepts}\label{sec:alg-geometry}

A subset $V$ of $\K^n$ is an \emph{algebraic set} if it is the set of common zeros of a finite collection of polynomials $g_1,g_2,\ldots,g_r$ with $g_i\in \K[x_1,x_2,\ldots,x_n]$ for $1\leq i\leq r$.

A finite union of algebraic sets is an algebraic set, and an arbitrary intersection of algebraic sets is again an algebraic set. Thus by taking the open subsets to be the complements of algebraic sets, we can define a topology, called the \emph{Zariski topology}, on $\K^n$.

A nonempty subset $V$ of a topological space $X$ is called \emph{irreducible} if it cannot be expressed as the union of two proper (Zariski) closed subsets. The empty set is not considered to be irreducible.
An \emph{affine algebraic variety} is an irreducible closed subset of some $\K^n$.

We define \emph{projective n-space}, denoted by $\mathbb{P}^n$, to be the set of
equivalence classes of $(n+1)$-tuples $(a_0,\ldots,a_n)$ of complex numbers, not all zero, under the equivalence relation given by
$(a_0,\ldots,a_n)\sim(\lambda a_0,\ldots,\lambda a_n)$ for all
$\lambda \in \K$, $\lambda\neq 0$.

A notion of algebraic variety may also be introduced in projective
spaces, giving the notion of a projective algebraic variety: a subset $V \subseteq \mathbb{P}^n$
is an \emph{algebraic set} if it is the set of common zeros of a finite
collection of homogeneous polynomials $g_1,g_2,\ldots,g_r$ with $g_i\in
\K[x_0,x_1,\ldots,x_n]$ for $1\leq i\leq r$. We call open subsets of irreducible projective varieties quasi-projective varieties.

For any integers $n$ and $d$, we define the Veronese map of degree $d$ as
\begin{align}
V_d\colon [x_0:x_1:\cdots:x_n]\mapsto [x_0^d:x_0^{d-1}x_1:\cdots:x_n^d]
\end{align}
where the notation with square brackets and colons denotes homogeneous coordinates and the expression in the output of $V_d$ ranges over all monomials of degree $d$ in $x_0,x_1,\ldots, x_n$. The image of $V_d$ is an algebraic variety called a \emph{Veronese variety}.

Finally, for an irreducible algebraic variety $V$, its \emph{$k$th secant variety} $\sigma_k(V)$ is  the Zariski closure of the union of subspaces spanned by $k$ distinct points chosen from $V$.

For more information about Veronese and secant varieties, refer to Example 2.4 and Example 11.30 in \cite{Harris92}.

\section{Quantum algorithm for polynomial interpolation}\label{sec:quantum-algorithm}

\subsection{The query model}\label{sec:query-model}

Using the standard concept of phase kickback, we encode the results of queries in the phase by performing standard queries in the Fourier basis.  We briefly explain these queries for the three types of fields we consider.

\subsubsection{Finite field $\Fq$}\label{sec:query-finite-field}

The order of a finite field can always be written as a prime power $q:=p^r$.  Let
$e\colon \Fq\to\C$ be the exponential function $e(z):=e^{i 2\pi\tr(z)/p}$ where the
trace function $\tr\colon \Fq\to\F_p$ is defined by
$\tr(z):=z+z^p+z^{p^2}+\cdots+z^{p^{r-1}}$. The Fourier transform over $\Fq$ is
a unitary transformation acting as
$\ket{x}\mapsto\frac{1}{\sqrt{q}}\sum_{y\in\Fq}e(xy)\ket{y}$ for all $x\in\Fq$.
The $k$-dimensional quantum Fourier transform (QFT) is given by $\ket{x}\mapsto
\frac{1}{q^{k/2}}\sum_{y\in\Fq^k}e(x\cdot y)\ket{y}$ for any $x\in\Fq^k$.

A phase query is simply the Fourier transform of a standard query.  By performing an inverse QFT, a query, and then a QFT, we map
$\ket{x,y} \mapsto e(yf(x))\ket{x,y}$
for any $x,y\in\Fq$.

As in the univariate case, our algorithm is nonadaptive, making all queries in parallel for a carefully chosen superposition of inputs.
With $k$ parallel queries, we generate a phase $\sum_{i=1}^k
{y_i}f(x_i)=\sum_{i=1}^k\sum_{j\in\J}{y_i}x_i^j c_j$ for the input $(x,y) \in \Fq^k \times \Fq^k$.
For convenience, we define $Z\colon \Fq^{nk}\times\Fq^k\to\Fq^J$ by $Z(x,y)_j=\sum_{i=1}^k
y_i{x_i}^j$ for $j\in\J$, so that $\sum_{i=1}^k {y_i}f(x_i)=Z(x,y)\cdot c$.

\subsubsection{Real numbers $\R$}\label{sec:query-real}

Let $e\colon \R\to\C$ be the exponential function
$e(x):=e^{i2\pi x}$. For any function $\psi$ whose Fourier transform exists, the QFT over $\R$ acts as
\begin{align}
\int_\R \dd x \, \psi(x)\ket{x} \mapsto \int_\R \dd y \, \Psi(y)\ket{y},
\end{align}
where $\Psi(y)=\int_{\R}\dd x \, e(-xy)\psi(x)$. By Parseval's theorem, the QFT is unitary.

As in the finite field case, we construct a phase query by making a standard query in the Fourier basis, giving
\begin{align}
\int_{\R^2} \dd x \, \dd y \,\psi(x,y)\ket{x,y}
& \mapsto
\int_{\R^2}\dd x \, \dd y \, e(yf(x)) \psi(x,y)\ket{x,y}.
\end{align}
An algorithm making $k$ parallel queries generates a phase $Z(x,y)\cdot c$, where we similarly define $Z\colon \R^{nk}\times\R^k\to\R^J$ by $Z(x,y)_j=\sum_{i=1}^k
y_i{x_i}^j$ for $j\in\J$.

\subsubsection{Complex numbers $\C$}\label{sec:query-complex}

The complex numbers can be viewed as a field extension of the real numbers of degree
2, namely $\C=\R[\sqrt{-1}]$.
For any positive integer $n$, let $\phi_n\colon \C^n\to\R^{2n}$ be an isomorphism
$\phi_n(x):=(\Re(x_1),\Im(x_1),\Re(x_2),\Im(x_2),\ldots,\Re(x_n),\Im(x_n))$, which
we also denote in boldface by $\vx$.
A complex number $x\in\C$ can be stored in a quantum register as
a tensor product of its real and imaginary parts,
$\ket{x}=\ket{\Re(x)}\ket{\Im(x)}$.

A complex function $\psi\colon \C^m\to\C^n$ can be seen as a function with $2m$ variables.
Let $\psi(x)=\tilde{\psi}(\vx)$.
By abuse of notation, we will neglect the tilde and write $\psi(x)=\psi(\vx)$.
Let $e\colon \C\to\C$ be the exponential function $e(x):=e^{i2\pi\Re(x)}$.
For any function $\psi\colon \C\to\C$ whose Fourier transform exists, we define
the transform
\begin{align}
\int_{\R^2}\dd^2\vx \, \psi(\vx)\ket{\vx}
\mapsto
\int_{\R^2}\dd^2\vy \, \Psi(\vy)\ket{\vy},
\end{align}
where $\Psi(\vy)=\int_{\R^2}\dd^2 \vx \, e(-\conj{y}x)\psi(\vx)$.
Note that in general $\Psi(\vy)$ cannot be written in the form of
$\Psi(y)$ with a complex variable $y\in\C$. To encode the output in the phase,
the queries act as
\begin{align}
&
\int_{\R^2}\dd^2\vx
\int_{\R^2}\dd^2\vy \,
\psi(\vx,\vy)
\ket{\vx,\vy} \nonumber\\
& \quad\mapsto
\int_{\R^2}\dd^2\vx
\int_{\R^2}\dd^2\vy
\int_{\R^2}\dd^2\vz \,
\psi(\vx,\vy)
e(-\conj{y}z)
\ket{\vx,\vz} \\
& \quad\mapsto
\int_{\R^2}\dd^2\vx
\int_{\R^2}\dd^2\vy
\int_{\R^2}\dd^2\vz \,
\psi(\vx,\vy)
e(-\conj{y}z)
\ket{\vx,\vz+\vf(x)} \\
& \quad\mapsto
\int_{\R^2}\dd^2\vx
\int_{\R^2}\dd^2\vy
\int_{\R^2}\dd^2\vz
\int_{\R^2}\dd^2\vu \,
\psi(\vx,\vy)
e(-\conj{y}z)e(\conj{u}(z+f(x)))
\ket{\vx,\vu} \\
& \quad\mapsto
\int_{\R^2}\dd^2\vx
\int_{\R^2}\dd^2\vy \,
\psi(\vx,\vy)
e(\conj{y}f(x))
\ket{\vx,\vy},
\end{align}
where we use the identity $\int_{\R^{2}}\dd^2\vy \,
e(y\conj{(x-x')})=\delta^{(2)}(\vx-\vx')$ for $x,x'\in\C$.

An algorithm making $k$ parallel queries generates a phase $\sum_{i=1}^k
\conj{y_i}f(x_i)=\sum_{i=1}^k\sum_{j\in\J} \conj{y_i}x_i^j c_j$.
We define $Z\colon \C^{nk}\times\C^k\to\C^J$ satisfying $Z(x,y)_j=\sum_{i=1}^k
y_i\conj{x_i}^j$ for $j\in\J$, so that $\sum_{i=1}^k \conj{y_i}f(x_i)=Z(x,y)\cdot c$.

\subsection{The algorithm}\label{sec:algorithm}

Our algorithm follows the same idea as in \cite{CvDHS16}: we perform $k$ phase queries in parallel for a carefully-chosen superposition of inputs, such that the output states corresponding to distinct polynomials are as distinguishable as possible.  For a $k$-query quantum algorithm, we consider the mapping $Z\colon \K^{nk}\times \K^k\to \K^{J}$ defined in \sec{query-model} for $\K=\Fq$, $\R$, and $\C$.  Reference \cite{CvDHS16} gave an optimal algorithm for $n=1$ using a uniform superposition over a unique set of preimages of the range $R_k:=Z(\K^{nk},\K^k)$ of $Z$, so we apply the same strategy here. For each $z\in R_k$, we choose a unique $(x,y)\in \K^{nk}\times \K^k$ such that $Z(x,y)=z$. Let $T_k$ be some set of unique representatives, so that $Z\colon T_k\to R_k$ is a bijection.

\subsubsection{$\K=\Fq$}\label{sec:alg-finite-field}

The algorithm generates a uniform superposition over $T_k$, performs $k$ phase
queries, and computes $Z$ in place, giving
\begin{align}
\frac{1}{\sqrt{|T_k|}}\sum_{(x,y)\in T_k}\ket{x,y}
&\mapsto
\frac{1}{\sqrt{|T_k|}}\sum_{(x,y)\in T_k}e(Z(x,y)\cdot c)\ket{x,y}
\label{eq:fq-alg-step3}
\mapsto
\frac{1}{\sqrt{|R_k|}}\sum_{z\in R_k}e(z\cdot c)\ket{z}.
\end{align}

We then measure in the basis of Fourier states $\ket{\widetilde c} := \frac{1}{\sqrt{q^J}}\sum_{z \in \Fq^J} e(z \cdot c)\ket{z}$.  A simple calculation shows that the result of this measurement is the correct
vector of coefficients with probability $|R_k|/q^J$.

\subsubsection{$\K=\R$}\label{sec:alg-real}

We consider a bounded subset $S\subseteq \R^J$ and a set $T_k'$ of unique
preimages of each element in $R_k\cap S$ such that $Z(T_k')=R_k\cap S$ and
$Z\colon T_k'\to R_k\cap S$ is bijective. The algorithm on input $\ket{\psi}$
with support $\supp(\psi)\subseteq R_k\cap S$ gives
\begin{align}
\ket{\psi}=\int_{R_k\cap S}\dd^J z \, \psi(z)\ket{z}
& \mapsto
\int_{R_k\cap S}\dd^J z \, \psi(z)\ket{z}|Z^{-1}(z)\rangle \\
& \mapsto
\int_{R_k\cap S}\dd^J z \, \psi(z)e(z\cdot c)\ket{z}|Z^{-1}(z)\rangle \\\label{eq:psi-r}
& \mapsto
\int_{R_k\cap S}\dd^J z \, \psi(z)e(z\cdot c)\ket{z}=:\ket{\psi_c}.
\end{align}

The choice of $S$ constrains the set of inputs that can be perfectly distinguished by this procedure, as captured by the following lemma.

\begin{lemma}[Orthogonality]\label{lem:orthogonality}
For positive integer $n$, let $m(A):=\int_A\dd^n z$ be the measure of the set $A\subseteq\R^n$.
Let $S$ be a bounded subset of $\R^n$ with nonzero measure.
Let $\ket{\widetilde{c}}=\frac{1}{\sqrt{m(S)}}\int_S\dd^n z \, e(c\cdot z)\ket{z}$ and let
$U$ be the maximal subset of $\R^n$ such that for any $c,c'\in U$ with $c\neq c'$,
\begin{align}
\langle\widetilde{c}'|\widetilde{c}\rangle
&= \frac{1}{m(S)}\int_{S}\dd^n z \, e((c-c')\cdot z) = 0.
\end{align}
Then there is a lattice $\Lambda$ such that $U\in\R^n/\Lambda$.
\end{lemma}
\begin{proof}
By definition, $c-c'$ must be a zero of the Fourier transform
$\FF(\I_S)$ of the indicator function $\I_S(z)$. We denote
$\Lambda:=\brc{c:\FF(\I_S)(c)=0}\cup\brc{0}$ and let $c_0\in U$.
Clearly $U\subseteq c_0+\Lambda$ as $\Lambda$ contains all zeros.
Since
$\langle\widetilde{c+c_0}|\widetilde{c_0}\rangle=0$ for all $c\in\Lambda\backslash\brc{0}$,
we have $c_0+\Lambda\subseteq U$ and $U=c_0+\Lambda$.
If $c\in\Lambda\backslash\brc{0}$, then
$\langle\widetilde{c_0+c}|\widetilde{c_0}\rangle=\langle\widetilde{c_0}|\widetilde{c_0-c}\rangle=0$
implies that $-c\in\Lambda$.
If $c,c'\in\Lambda\backslash\brc{0}$, then
$\langle\widetilde{c+c_0}|\widetilde{-c'+c_0}\rangle=\langle\widetilde{c+c'+c_0}|\widetilde{c_0}\rangle=0$
implies $c+c'\in\Lambda\backslash\brc{0}$. Therefore $\Lambda$ is an additive
subgroup of $\R^n$.

Now we prove that $\Lambda$ is a lattice.
For $\epsilon>0$, $\delta\in B(\epsilon)$, and $c\in\Lambda$,
\begin{align}
|\langle\widetilde{c+\delta}|\widetilde{c}\rangle|^2=
\abs{\int_S \dd^n z \, e(\delta\cdot z)}^2
\geq \abs{\int_S \dd^n z \, \cos(2\pi\delta\cdot z)}^2> 0
\end{align}
if $S\subseteq B(r)$ for $r<\frac{1}{4\epsilon}$.
Thus $B(\epsilon)$ contains exactly one element in $\Lambda$ and hence
$\Lambda$ is discrete.
\end{proof}
Roughly speaking, \lem{orthogonality} is a consequence of the uncertainty
principle: restricting the support to a finite window limits the precision with which we can determine the Fourier transform.  In the proof, note that a larger window offers
better resolution of the coefficients.

We have shown that the set $\Lambda$ of perfectly distinguishable coefficients
forms a lattice.  We also require the set $\brc{\ket{\widetilde{c}}:c\in\Lambda}$
to be a complete basis.  Since $\bkt{z|\widetilde{c}}=\frac{1}{\sqrt{m(S)}}e(z\cdot
c)$, completeness implies that $\ket{z}$ is of the form
$\sum_{c\in\Lambda}e(-z\cdot c)\ket{\widetilde{c}}$ up to a normalization constant.
More formally, we have the following lemma.

\begin{lemma}[Completeness]\label{lem:completeness}
For positive integer $n$, let $m(A):=\int_A\dd^n z$ be the measure of the set $A\subseteq\R^n$.
Let $\Lambda$ be a discrete additive subgroup of $\R^n$.
Let $S$ be a bounded set with nonzero measure and
$\ket{\widetilde{c}}=\frac{1}{\sqrt{m(S)}}\int_S\dd^n z \, e(z\cdot c)\ket{z}$.
Then $\brc{\ket{\widetilde{c}}:c\in\Lambda}$ forms a complete basis over support $S$
if and only if $S$ is a fundamental domain of the dual lattice of $\Lambda$.
\end{lemma}
\begin{proof}
Let $\widetilde{\Lambda}$ be the dual lattice of $\Lambda$.
We observe that (ignoring the normalization constant)
\begin{align}
\sum_{c\in\Lambda}e(-z\cdot c)\ket{\widetilde{c}}
&= \int_S\dd^J z' \sum_{c\in\Lambda} e((z'-z)\cdot c)\ket{z'}
\label{eq:dirac-comb}
= \int_S\dd^J z' \sum_{z_0\in\widetilde{\Lambda}} \delta(z'-z-z_0)\ket{z'} \\
&=\sum_{z_0\in\widetilde{\Lambda}}\I_S(z+z_0)\ket{z+z_0}
=\ket{(z+\widetilde{\Lambda})\cap S}.
\end{align}
In \eq{dirac-comb}, $\sum_{c\in\Lambda}e(z\cdot c)=\sum_{z_0\in\widetilde{\Lambda}}\delta(z-z_0)$
up to a constant factor \cite[Section 7.2]{Hor83}.
The set $(z+\widetilde{\Lambda})\cap S$ cannot be empty, so a fundamental domain of
$\widetilde{\Lambda}$ is a
subset of $S$.
For $z,z'\in\R^n$,
$\langle(z+\widetilde{\Lambda})\cap S|(z'+\widetilde{\Lambda})\cap S\rangle=0$
if $z'\notin z+\widetilde{\Lambda}$, which implies that $S$ is a subset of a fundamental domain
of $\widetilde{\Lambda}$.
\end{proof}
\lem{completeness} further restricts the bounded set $S$ has to be a fundamental
region of $\widetilde{\Lambda}$.  Without loss of generality, one may choose $S$ to
be a fundamental domain of a lattice centered at zero. In the last step, the
algorithm applies the unitary operator
\begin{align}\label{eq:unitary}
\frac{1}{\sqrt{m(S)}}\sum_{c'\in\Lambda}\int_{S}\dd^J z \, e(-z\cdot c')\ket{c'}\bra{z}
\end{align}
to the state $\ket{\psi_c}$ in \eq{psi-r}.
The algorithm outputs $c'\in\Lambda$ with probability
\begin{align}\label{eq:success-prob-r}
\frac{1}{m(R_k\cap S)m(S)}\abs{\int_{R_k\cap S}\dd^J z \, \psi(z)e(z\cdot (c-c'))}^2
\leq \frac{m(R_k\cap S)}{m(S)},
\end{align}
where the upper bound follows from the Cauchy-Schwarz inequality.
The maximum is reached if $\psi(z)=\frac{1}{\sqrt{m(R_k\cap S)}}\I_{R_k\cap S}(z)$ and
$c$ happens to be a lattice point.  If $c\notin\Lambda$, the algorithm returns
the closest lattice point with high probability.

To achieve arbitrarily high precision, one may want to take $S\to\R^J$.  In
this limit, the basis of coefficients is normalized to the Dirac
delta function, i.e., $\bkt{\widetilde{c}'|c}=\delta^{(J)}(c-c')$.
In this case, $\Lambda\to\R^J$ and the unitary operator in \eq{unitary} becomes
the $J$-dimensional QFT over the real numbers.
However, for the interpolation problem, the success probability
$\frac{m(R_k\cap S)}{m(S)}$ is not well-defined in the limit $S\to\R^J$ since
different shapes for $S$ can give different probabilities.
Thus it is necessary to choose a bounded region, and we leave the optimal choice as an open question.

Though the size of the fundamental domain $S$ affects the resolution of the
coefficients, it does not affect the maximal success probability
$\frac{m(R_k\cap S)}{m(S)}$.  This can be seen by scale invariance: for every
$z\in R_k$, there is a preimage $(x,y)$ such that $Z(x,y)=z$. Then $\lambda
z\in R_k$ since $Z(x,\lambda y)=\lambda z$ for any $\lambda\in\R$.  In terms of the bijection $\ell\colon
z\mapsto \lambda z$ for $\lambda\in\R^\times$, we have
$\ell(R_k)=R_k$ and $\ell(R_k\cap S)=R_k\cap\ell(S)$.  Then
$m(R_k\cap\ell(S))=m(\ell(R_k\cap S))=\lambda^{J}m(R_k\cap S)$ and hence
$\frac{m(R_k\cap \ell(S))}{m(\ell(S))}=\frac{m(R_k\cap S)}{m(S)}$.
Thus we can make the precision arbitrarily high by taking $S$ arbitrarily
large, and we call $\frac{m(R_k \cap S)}{m(S)}$ the success probability of the
algorithm.

\subsubsection{$\K=\C$}\label{sec:alg-complex}
We consider a bounded set $S\subseteq \C^J$ and a set $T_k'$ of unique
preimages of each element in $R_k\cap S$ such that $Z(T_k')=R_k$ and $Z\colon T_k'\to
R_k\cap S$ is bijective.
The algorithm on input $\ket{\psi}$ with support $\supp(\psi)\subseteq R_k\cap S$ gives
\begin{align}
\ket{\psi}=
\int_{\phi(R_k\cap S)}\dd^{2J}\vz \,
\psi(\vz)\ket{\vz}
& \mapsto
\int_{\phi(R_k\cap S)}\dd^{2J}\vz \,
\psi(\vz)\ket{\vz}\vert \phi(Z^{-1}(z))\rangle \\
& \mapsto
\int_{\phi(R_k\cap S)}\dd^{2J}\vz \,
\psi(\vz)
e(z\cdot c)
\ket{\vz}\vert \phi(Z^{-1}(z))\rangle \\\label{eq:psi-c}
& \mapsto
\int_{\phi(R_k\cap S)}\dd^{2J}\vz \,
\psi(\vz)
e(z\cdot c)
\ket{\vz} =:\ket{\psi_{\vc}}.
\end{align}

By \lem{orthogonality} and \lem{completeness},
the set $S$ must be a fundamental domain in $\C^{J}$.
Let $\brc{\ket{\widetilde{\vc}}:c\in\Lambda}$ be the measurement basis.
In the last step of the algorithm,
we apply the unitary operator
\begin{align}
\frac{1}{\sqrt{m(S)}}
\sum_{\vc'\in\phi(\Lambda)}
\int_{\phi(S)}\dd^{2J} \vz \, e(-z\cdot c')\ket{\vc'}\bra{\vz}
\end{align}
to the state $\ket{\psi_\vc}$ in \eq{psi-c}.
The algorithm outputs $c'\in\Lambda$ with probability
\begin{align}\label{eq:success-prob-c}
\frac{1}{m(R_k\cap S)m(S)}
\abs{\int_{\phi(R_k\cap S)}\dd^{2J}\vz \, \psi(\vz)e(z\cdot (c-c'))}^2.
\end{align}
Again, since $\ket{\psi}$ is normalized, \eq{success-prob-c} cannot be arbitrarily large.
By the Cauchy-Schwarz inequality, \eq{success-prob-c} is
upper bounded by $\frac{m(R_k\cap S)}{m(S)}$; this maximal success probability is obtained if $\psi(\vz)=\frac{1}{\sqrt{m(R_k\cap
S)}}\I_{\phi(R_k\cap S)}(\vz)$ and $c$ happens to be a lattice point.  If
$c\notin\Lambda$, the algorithm returns the closest lattice point with high
probability.

By the same argument as in \sec{alg-real}, we can show scale invariance holds
for complex numbers: for $\ell\colon z\mapsto \lambda z$ where
$z\in\C^J$ and $\lambda\in\R^\times$,
$\frac{m(R_k\cap S)}{m(S)}=\frac{m(R_k\cap\ell(S))}{m(\ell(S))}$.
Thus we can make the precision of the algorithm arbitrarily high by taking
$S$ arbitrarily large without affecting the maximal success probability.

\subsection{Performance}\label{sec:perf}

We have shown in \sec{alg-finite-field} that the optimal success probability is at most $|R_k|/q^J$ for $\K=\Fq$. For real and complex numbers, we consider a bounded support $S$ in which the algorithm is performed. The success probability of the algorithm with this choice is at most $\frac{m(R_k\cap S)}{m(S)}$, as shown in \eq{success-prob-r} and \eq{success-prob-c}. To establish the query complexity, first we show that if $\dim R_k=J$, the algorithm outputs the coefficients with bounded error.

\begin{lemma}\label{lem:bounded-error}
For positive integers $n,k,d$, let $J:=\binom{n+d}{d}$ and let
$m(A):=\int_A\dd^J z$ be the volume of $A\subseteq R^{J}$.
Let $Z\colon \K^{nk}\times\K^k$, $Z(x,y)=\sum_{i=1}^k y_ix_i^j$ for an infinite field $\K$.
Let $R_k=Z(\K^{nk},\K^k)$ be the range of $Z$.
If $\dim R_k=J$, then $\frac{m(R_k\cap S)}{m(S)}>0$ if $S$ is a fundamental domain centered at $0$.
\end{lemma}

\begin{proof}
$R_k$ is a constructible set for $\K=\C$ and it is a semialgebraic set for $\K=\R$.
By \cite{BBO15} and \cite{BT15}, $R_{k}$ has non-empty interior if $\dim(R_k)=J$ for both cases.

$S$ is a fundamental domain centered at $0$ with finite measure, so we only need to prove that $m(R_k\cap S)$ is of positive measure, or equivalently, that the interior of $R_k$ and the interior of $S$ have non-empty intersection.

If this is not the case, then any interior point of $S$ cannot be in the interior of $R_k$. By scale invariance of $R_k$, any point in $\K^n$ except $0$ cannot be in the interior of $R_k$, which contradicts the fact that $R_k$ has non-empty interior given $\dim(R_k)=J$.
\end{proof}

\lem{bounded-error} shows that for infinite fields, although we perform the
algorithm over a bounded support, the query complexity can be understood by
considering the dimension of the entire set $R_k$.
Moreover, by invoking recent work on typical ranks, we can establish the
minimum number of queries to determine the coefficients almost surely.

Now let $v_d(x_1,x_2,\ldots,x_n)$ be the ${n+d\choose d}$-dimensional vector that contains all monomials with variables $x_1,\ldots,x_n$ of degree no more than $d$ as its entries. Let
\begin{align}
	X_{n,d}:=\{v_d(x_1,x_2,\ldots,x_n):x_1,x_2,\ldots,x_n \in \K\}
\end{align}
where $\K$ is a given ground field. For example, we have
\begin{align}
	X_{3,2}=\{(x_1^2,x_2^2,x_3^2,x_1x_2,x_1x_3,x_2x_3,x_1,x_2,x_3,1)^T:x_1,x_2,x_3\in \K\}.
\end{align}
Our question is to determine the smallest number $k$ such that a generic vector in $\K^{n+d\choose d}$ can be written as a linear combination of no more than $k$ elements from $X_{n,d}$. More precisely, we have $R_k=\{\sum_{i=1}^k c_i v_i: c_i\in \K, v_i\in X_{n,d}\}$, and we ask what is the smallest number $k$ such that $R_k$ has full measure in $\K^{n+d\choose d}$.

Our approach requires basic knowledge of algebraic geometry---specifically, the concepts of Zariski topology, Veronese variety, and secant variety. Formal definitions can be found in \sec{alg-geometry}. For the reader's convenience, we also explain these concepts briefly when we first use them.

Now we make two simple observations.
\begin{enumerate}
\item In general, $v_d(x_1,x_2,\ldots,x_n)$ can be treated as an ${n+d\choose d}$-dimensional vector that contains all monomials with variables $x_1,\ldots,x_n,x_{n+1}$ of degree $d$ as its entries, by simply taking the map $(x_1,x_2,\ldots,x_n)\mapsto (\frac{x_1}{x_{n+1}},\ldots,\frac{x_n}{x_{n+1}})$ and multiplying by $x_{n+1}^d$. For example, applying this mapping to $X_{3,2}$ gives
\begin{equation*}
X_{3,2}'=\{(x_1^2,x_2^2,x_3^2,x_1x_2,x_1x_3,x_2x_3,x_1x_4,x_2x_4,x_3x_4,x_4^2)^T:x_1,x_2,x_3,x_4\in \K\}.
\end{equation*}
The new set $X_{n,d}'$ is slightly bigger than $X_{n,d}$ since it also contains those points corresponding to $x_{n+1}=0$, but this will not affect our calculation since the difference is just a measure zero set in $X_{n,d}'$.
\item The set $X_{n,d}'$ is the Veronese variety. One may also notice that this set is isomorphic to $\big((x_1,x_2,\ldots,x_{n+1})^T\big)^{\otimes d}$ in the symmetric subspace.
\end{enumerate}

These observations imply that instead of studying $R_k$, we can study the new set
\begin{align}
  R_k'=\brc{\sum\limits_{i=1}^k c_i v_i': c_i\in \K, v_i'\in X'_{n,d}}.
\end{align}
In general, we have a sequence of inclusions:
\begin{align}
X_{n,d}'=R_1'\subseteq R_2' \subseteq \cdots \subseteq R_k' \subseteq \cdots \subseteq \K^{n+d \choose d}.
\end{align}
By taking the Zariski closure, we also have
\begin{align}
	\overline{X_{n,d}'}=\overline{R_1'}\subseteq \overline{R_2'} \subseteq \cdots \subseteq \overline{R_k'} \subseteq \cdots \subseteq \K^{n+d \choose d}
\end{align}
where $\overline{R_k'}$ is the $k$th secant variety of the Veronese variety $X_{n,d}'$.

Palatini showed the following \cite{P09,A87}:
\begin{lemma}
If $\dim \overline{R_{k+1}'} \leq \dim \overline{R_k'}+1$, then $\overline{R_{k+1}'}$ is linear.
\end{lemma}

In particular, this shows that if $\dim\overline{R_k'}={n+d\choose d}$, then $\overline{R_k'}={\K}^{n+d\choose d}$.

For an infinite field $\K$, define $k_{\K}$ to be the smallest integer such that $\frac{m(R_{k_{\K}}\cap S)}{m(S)}=1$.  Thus $k_{\K}$ represents the minimal number of queries such that our algorithm succeeds with probability $1$. For the finite field case $\K=\mathbb{F}_q$, we only require that $\frac{m(R_{k_{\mathbb{F}_q}}\cap S)}{m(S)}$ goes to $1$ when $q$ tends to infinity.

\subsubsection{$\K=\C$}\label{sec:performance-complex}
A theorem due to Alexander and Hirschowitz \cite{AH95}
implies an upper bound on the query complexity of polynomial interpolation over $\C$.
\begin{theorem}[Alexander-Hirschowitz Theorem, \cite{AH95}]
The dimension of $\overline{R_k'}$ satisfies
\begin{align}
	\dim \overline{R_k'}=
\begin{cases} k(n+1)-\frac{k(k-1)}{2} & \textrm{$d=2, 2\leq k\leq
n$;}\\ \binom{n+d}{d}-1 &
\textrm{$(d,n,k)=(3,4,7),(4,2,5),(4,3,9),(4,4,14)$;}\\
\min\{k(n+1),\binom{n+d}{d}\}  & \textrm{otherwise.} \end{cases}
\end{align}
\end{theorem}
\noindent Thus, the minimum $k$ to make
$\overline{R_k'}=\mathbb{C}^{\binom{n+d}{d}}$ is
\begin{align}
k_{\mathbb{C}}(n,d):= \begin{cases} n+1 & \textrm{$d=2, n\geq 2$;}\\
\lceil\frac{1}{n+1}{\binom{n+d}{d}}\rceil+1 &
\textrm{$(n,d)=(4,3),(2,4),(3,4),(4,4)$;}\\
\lceil\frac{1}{n+1}{\binom{n+d}{d}}\rceil  & \textrm{otherwise.} \end{cases}
\end{align}
By parameter counting, we see that $R_k$ is of full measure in $R_k'$.
It remains to show that $R_k'$ is of full measure in its Zariski closure $\overline{R_k'}$:

\begin{theorem}\label{thm:rk'}
$R_k'$ is of full measure in $\overline{R_k'}$.
\end{theorem}

\begin{proof}
$R_k'$ is just the image of the map
$(Q_1,Q_2,\ldots,Q_k)\mapsto (Q_1+Q_2+\cdots+Q_k)$. By Exercise $3.19$ in Chapter II of
\cite{H77}, $R_k'$ is a constructible set, so it contains an
open subset of each connected component of $\overline{R_k'}$. Therefore its complement is of
measure $0$.
\end{proof}

This immediately implies the following:

\begin{corollary}
$R_k$ has measure $0$ in $\mathbb{C}^{n+d\choose d}$ for $k<k_{\mathbb{C}}(n,d)$ and measure $1$ in $\mathbb{C}^{n+d\choose d}$ for $k\geq k_{\mathbb{C}}(n,d)$.
\end{corollary}

Thus, as the integer $k$ increases, $\frac{m(R_k'\cap S)}{m(S)}$ suddenly jumps from $0$ to $1$ at the point $k_{\mathbb{C}}(n,d)$, and so does $\frac{m(R_k\cap S)}{m(S)}$. This implies part (3) of \thm{main}.

\subsubsection{$\K=\R$}\label{sec:performance-real}

Now consider the case $\K=\R$. For $d=2$, $(n+1)$-variate symmetric tensors are simply $(n+1) \times (n+1)$ symmetric matrices, so a random $(n+1)$-variate symmetric tensor will be of rank $n+1$ with probability $1$. However, if the order of the symmetric tensors is larger than $2$, the situation is much more complicated. For example, a random bivariate symmetric tensor of order $3$ will be of two different ranks, $2$ and $3$, both with positive probabilities.

From the perspective of algebraic geometry, it still holds that $\overline{R_k'}=\mathbb{R}^{n+d \choose d}$ for $k\geq k_{\C}(n,d)$, and for $k<k_{\mathbb{C}}(n,d)$, $\overline{R_k'}$ is of measure zero in $\R^{n+d\choose d}$. It also holds that $R_k$ is of full measure in $R_k'$. However, the claim that $R_k'$ has full measure in $\overline{R_k'}$ no longer holds over $\R$. As we explained in the proof of \thm{rk'}, $R_k'$ is  the image of the map
$(Q_1,Q_2,\ldots,Q_k)\mapsto (Q_1+Q_2+\cdots+Q_k)$. For an algebraically closed field $\K$, it is known that the image of any map is always a constructible set in its Zariski closure. Thus $R_k'$ is of full measure in $\overline{R_k'}$. Over $\R$, it is easy to verify that the image may not be of full measure in its Zariski closure (a simple counterexample is $x\mapsto x^2$). Consequently, over $\C$, $R_k'$ has non-empty interior for a unique value of $k$, and this value of $k$ is called the \emph{generic rank}. Over $\R$, $R_k'$ is just a semialgebraic set and it has non-empty interior for several values of $k$, which are called the \emph{typical ranks}.

For the univariate case, we have the following theorem:

\begin{theorem}[\cite{CO12,CR11}]
For $n=1$, all integers from $k_{\mathbb{C}}=\lceil \frac{d+1}{2}\rceil$ to $k_{\R}=d$ are typical ranks.
\end{theorem}

For the multivariate case $n\geq 2$, it still holds that $k_{\mathbb{C}}(n,d)$ defined in \sec{performance-complex} is the smallest typical rank \cite{BT15}.  According to \cite{BBO15}, every rank between $k_{\mathbb{C}}(n,d)$ and the top typical rank $k_{\mathbb{R}}(n,d)$ is also typical. Thus we only need to study the top typical rank $k_{\mathbb{R}}(n,d)$. Unfortunately, the top typical rank in general is not known. In the literature, considerable effort has been devoted to understanding the maximum possible rank $k_{\max}(n,d)$, which, by definition, is also an upper bound for $k_{\mathbb{R}}(n,d)$.
In particular, we have $k_{\max}(n,2)\leq n+1$ for $n \ge 2$, $k_{\max}(2,4)\leq 11$, $k_{\max}(3,4)\leq 19$, $k_{\max}(4,4)\leq 29$, $k_{\max}(4,3)\leq 15$, and $k_{\max}(n,d)\leq 2\lceil\frac{1}{n+1}{n+d\choose d}\rceil$ otherwise \cite{BT15}.

The above result implies $k_{\mathbb{R}}(n,d)\leq k_{\max}(n,d)\leq 2 k_{\mathbb{C}}(n,d)$. We also mention a few other upper bounds on $k_{\max}(n,d)$. Trivially we have $k_{\max}(n,d)\leq {n+d \choose d}$. In \cite{G96, LT10}, this was improved to $k_{\max}(n,d)\leq {n+d \choose d}-n$. Later work showed that $k_{\max}(n,d)\leq {n+d-1\choose n}$ \cite{BS08}. Jelisiejew then proved that $k_{\max}(n,d)\leq {n+d-1\choose n}-{n+d-5 \choose n-2}$ \cite{J13}, and Ballico and De Paris then improved this to $k_{\max}(n,d)\leq {n+d-1\choose n}-{n+d-5 \choose n-2}-{n+d-6 \choose n-2}$ \cite{BD13}. For small cases, these bounds may be stronger than the bound $k_{\max}(n,d)\leq 2 k_{\mathbb{C}}(n,d)$ mentioned above.

To summarize, we have the following, which implies part (2) of \thm{main}:
\begin{theorem}
As the integer $k$ increases from $k_{\mathbb{C}}(n,d)-1$ to $k_{\mathbb{R}}(n,d)\leq 2k_{\mathbb{C}}(n,d)$, $\frac{m(R_k'\cap S)}{m(S)}$ forms a strictly increasing sequence from $0$ to $1$, and so does $\frac{m(R_k\cap S)}{m(S)}$.
\end{theorem}

\subsubsection{$\K=\Fq$} \label{sec:performance-finite-field}

We link the finite field case with the complex case using the Lang-Weil theorem:

\begin{theorem}[Lang-Weil Theorem, \cite{L54}]\label{thm:lang}
There exists a constant $A(n,d,r)$ depending only on $n, d, r$ such that for any variety $V \subseteq \mathbb{P}^n$ with dimension $r$ and degree $d$, if we define $V$ over a finite field $\Fq$, the number of points in $V$ must satisfy
\begin{align}
\vert N-q^r\vert \leq (d-1)(d-2) q^{r-\frac{1}{2}}+A(n,d,r) q^{r-1}.
\end{align}
\end{theorem}

The Lang-Weil theorem shows that when $q$ is large enough, the number of points in a variety over $\Fq$ is very close to $q^{\dim V}$.  So it actually tells us that $\frac{m(R_k'\cap S)}{m(S)}=0$ if $k<k_{\mathbb{C}}(n,d)$. It remains unclear whether $\frac{m(R_k'\cap S)}{m(S)}>0$ for $k=k_{\mathbb{C}}(n,d)$. Once again, for the finite field case, when we talk about the measure, we always assume $q$ is sufficiently large. As in the real field case,  the main challenge now is to study the measure of $R_k'$ in $\overline{R_k'}$.

For the upper bound, recall our notation that $v_d(x_1,x_2,\ldots,x_n)$ is the ${n+d\choose d}$-dimensional vector that contains all monomials with degree no more than $d$ as its entries.

Here we make a slight change to the definition in which we require all those $x_i$s in $v_d$ to be nonzero. We can similarly define $X_{n,d}''$ and $R_k''$.
We prove the following:

\begin{lemma}\label{lem:upper}
Let $r_{n,d}$ be the minimum number such that $\vert R_{r_{n,d}}''\vert=q^{n+d\choose d}-O(q^{{n+d\choose d}-1})$.  Then $r_{n,d}\leq r_{n-1,d}+ r_{n,d-1}$.
\end{lemma}
\begin{proof}
The proof is by induction on $n+d$.

For $n+d=2$, it is easy to verify $r_{2,2}=3\leq r_{1,2}+r_{2,1}=2+1$. Assume \lem{upper} holds for $n+d\leq m-1$ and consider the pair $(n,d)$ with $n+d=m$. For the sake of readability, we first explain how the induction works for the specific example $(n,d)=(3,2)$, and then generalize our idea to any $(n,d)$.

The vector
\begin{align}
v_2(x_1,x_2,x_3)=(x_1^2,x_2^2,x_3^2,x_1x_2,x_1x_3,x_2x_3,x_1,x_2,x_3,1)^T \in X_{3,2}''
\end{align}
can be rearranged as $(x_3,x_3x_1,x_3x_2,x_3^2,x_1^2,x_2^2,x_1x_2,x_1,x_2,1)^T$. The first $4$ entries can be rewritten as $x_3^2(\frac{1}{x_3},\frac{x_1}{x_3},\frac{x_2}{x_3},1)^T = x_3^2 v_1(\frac{1}{x_3},\frac{x_1}{x_3},\frac{x_2}{x_3})$, and the last $6$ entries form $v_2(x_1,x_2)$.

When $(x_1,x_2,x_3)$ ranges over all $3$-tuples in $\mathbb{F}_q\setminus \{0\}$, $(\frac{1}{x_3},\frac{x_1}{x_3},\frac{x_2}{x_3})$ also ranges over all possible $3$-tuples in $\mathbb{F}_q\setminus \{0\}$. By assumption, if we take linear combinations of $r_{3,1}$ vectors chosen from $X_{3,2}''$, the first $4$ entries will range over no fewer than $q^{3+1\choose 1}-O(q^{{3+1\choose 1}-1})$ different vectors in $\mathbb{F}_q^{3+1\choose 1}$.

For any such linear combination, we can add $r_{2,2}$ extra vectors from $X_{3,2}$ with the restriction that $x_3=0$, which will guarantee these extra vectors do not affect the first $4$ entries. By assumption, the last $6$ entries will range over no fewer than $q^{2+2\choose 2}-O(q^{{2+2\choose 2}-1})$ different vectors in $\mathbb{F}_q^{2+2\choose 2}$.

Thus, in total, we have $\smash{\bigl(q^{3+1\choose 1}-O(q^{{3+1\choose 1}-1})\bigr)\bigl(q^{2+2\choose 2}-O(q^{{2+2\choose 2}-1})\bigr)}$ different vectors in $\mathbb{F}_q^{3+2\choose 2}$ if we take linear combinations of $r_{3,1}+r_{2,2}$ vectors from $X_{3,2}''$, which implies $r_{3,2}\leq r_{3,1}+r_{2,2}$.

For general $(n,d)$, the analogous partition of $v_d(x_1,x_2,\ldots,x_n)$ is still valid. Those ${n+d \choose d}-{n-1+d \choose d}={n+d-1 \choose d-1}$ entries involving $x_n$ will form $x_n^{d-1} v_{d-1}(\frac{1}{x_n},\frac{x_1}{x_n},\ldots,\frac{x_{n-1}}{x_n})$ and the rest will form $v_{d}(x_1,x_2,\ldots,\allowbreak x_{n-1})$. All arguments follow straightforwardly, so we have $r_{n,d}\leq r_{n-1,d}+r_{n,d-1}$ for $n+d=m$ and for any $(n,d)$ by induction.
\end{proof}

\begin{corollary}\label{cor:upper}
$r_{n,d}\leq {n+d-1\choose d-1}$.
\end{corollary}

\begin{proof}
We use induction on $n+d$. For $n+d=2$, it is easy to verify. If it is true for $n+d=m$, then for $n+d=m+1$, we have $r_{n,d}\leq r_{n-1,d}+r_{n,d-1}\leq {n+d-2 \choose d-1}+{n+d-2 \choose d-2} ={n+d-1 \choose d-1}$.
\end{proof}

$R_k''$ is obviously a subset of $R_k'$, so $k_{\mathbb{F}_q}(n,d)\leq  r_{n,d}$. By combining \thm{lang} and \cor{upper}, we have the following, which implies part (1) of \thm{main}:

\begin{corollary}
$k_{\mathbb{C}}(n,d)\leq k_{\mathbb{F}_q}(n,d)\leq  r_{n,d}\leq {n+d-1\choose d-1} = \frac{d}{n+d}\binom{n+d}{d}$.
\end{corollary}

\begin{remark}
By combining \cor{upper} with the fact $r_{n,2}\geq k_{\mathbb{C}}(n,2)$, we have $r_{n,2}=n+1$.
\end{remark}

\begin{remark}
It was previously known that $r_{n,1}=1$ \cite{BV97,BCW02} and $r_{1,d}=\lceil\frac{d+1}{2}\rceil$ \cite{CvDHS16}.  We can further refine the upper bound using these boundary conditions:
\begin{align}
r_{n,d}&\leq \sum\limits_{i=0}^{d-2}{n-2+i \choose i} r_{1,d-i}+{d+n-3\choose d-1}\leq \sum\limits_{i=0}^{d-2}{n-2+i \choose i} \frac{d-i+3}{2}+{d+n-3\choose d-1} \nonumber\\
&= \frac{n+d+2}{2} {n+d-3 \choose n-1}-\frac{n-1}{2}{n+d-2\choose n}+{d+n-3\choose d-1}.
\end{align}
\end{remark}

\section{Optimality}\label{sec:opt}

In this section, we show that our algorithm is optimal for the case of finite fields. Specifically, we show that no $k$-query quantum algorithm can succeed with probability greater than $|R_k|/q^J$. This follows by essentially the same argument as in the univariate case \cite{CvDHS16}.

First we show that the final state of a $k$-query algorithm is restricted to a subspace of dimension $|R_k|$.  We prove the following:

\begin{lemma}[{cf.\ Lemma 3 of \cite[arXiv version]{CvDHS16}}]\label{lem:dim-span}
Let $J:=\binom{n+d}{d}$, and let $\ket{\psi_c}$ be the state of any quantum algorithm after $k$ queries, where the black box contains $c\in\Fq^J$.
Then $\dim\spn\{\ket{\psi_c}:c\in\Fq^J\}\leq |R_k|$.
\end{lemma}
\begin{proof}
Following the same technique as in the proof of Lemma 3 in \cite[arXiv version]{CvDHS16}, consider a general
$k$-query quantum algorithm $U_kQ_cU_{k-1}Q_c\ldots Q_cU_1Q_cU_0$ acting on
a state space of the form $\ket{x,y,w}$ for an arbitrary-sized workspace register $\ket{w}$.
Here $Q_c\colon \ket{x,y}\mapsto e(yf(x))\ket{x,y}$ for $x \in \Fq^n, y\in\Fq$ is the phase query.
Starting from the initial state $\ket{x_0,y_0,w_0}=\ket{0,0,0}$, we can write
the output state in the form
\begin{align}\label{eq:output-state}
\ket{\psi_c}=\sum_{z\in R_k}
e(z\cdot c)\ket{\xi_z},
\end{align}
where with $x=(x_1,\ldots,x_k)\in(\Fq^{n})^k$, $y=(y_1,\ldots,y_k)\in\Fq^k$,
$w=(w_1,\ldots,w_{k+1})$ and $I$ an appropriate index set,
\begin{align}
\ket{\xi_z}=\sum_{(x,y)\in Z^{-1}(z)}
\sum_{\substack{x_{k+1} \in \Fq^n,\\y_{k+1}\in\Fq,\\ w\in I^{k+1}}}
\prn{\prod_{j=0}^{k}\bkt{x_{j+1},y_{j+1},w_{j+1}\abs{U_j}x_j,y_j,w_j}}
\ket{x_{k+1},y_{k+1},w_{k+1}}.
\end{align}
Then $\dim\spn\{\ket{\psi_c}:c\in\Fq^J\}\leq\dim\spn\brc{\ket{\xi_z}:z\in R_k}\leq|R_k|$.
\end{proof}
We also use the following basic lemma about the distinguishability of a set of quantum states in a space of restricted dimension.
\begin{lemma}[Lemma 2 of {\cite[arXiv version]{CvDHS16}}]\label{lem:optimality}
Suppose we are given a state $\ket{\phi_c}$ with $c\in C$ chosen uniformly at random.  Then the probability of correctly determining $c$ with some orthogonal measurement is at most $\dim\spn\{\phi_c:c\in C\}/|C|$.
\end{lemma}
Combining these lemmas, the success probability of multivariate interpolation under the uniform distribution over $c \in \Fq^J$ (and hence also in the worst case) is at most $|R_k|/q^J$.

Unfortunately it is unclear how to generalize this argument to the infinite-dimensional case, so we leave lower bounds on the query complexity of polynomial interpolation over $\R$ and $\C$ as a topic for future work.

\section{Conclusion and open problems}\label{sec:open}

In this paper, we studied the number of quantum queries required to determine the coefficients of a degree-$d$ polynomial in $n$ variables over a field $\K$. We proposed a quantum algorithm that works for $\K=\C$, $\R$, or $\Fq$, and we used it to give upper bounds on the quantum query complexity of multivariate polynomial interpolation in each case. Our results show a substantially larger gap between classical and quantum algorithms than the univariate case.

There are still several open questions that remain.  Recall that $k_{\K}$ represents the minimal number of queries required for our algorithm to succeed with probability $1$ over the field $\K$ (or with probability approaching $1$ for large $q$ if $\K=\Fq$).  First, for the finite field case $\K=\Fq$, can we bound $k_{\mathbb{F}_q}$ by $C k_{\mathbb{C}}$ where $C$ is a constant independent of the degree $d$? For the values of $(n,d)$ for which explicit values of $k_{\C}$, $k_{\R}$ and $k_{\Fq}$ are known, we always have $k_{\C}\leq k_{\Fq}\leq k_{\R}$. For example, $k_{\C}(1,d)=k_{\Fq}(1,d)=\lceil \frac{d+1}{2}\rceil \leq d=k_{\R}(1,d)$ and  $k_{\C}(n,2)=k_{\Fq}(n,2)=k_{\R}(n,2)=n+1$.  Thus it is plausible to conjecture that $k_{\Fq}(n,d)\leq k_{\R}(n,d)$, which would imply $k_{\Fq}\leq 2k_{\C}$.

Another question is whether we can always obtain positive success probability with only $k_\mathbb{C}$ queries. We know that $k_{\C}$ queries are sufficient to achieve positive success probability for $\K=\C$ and $\R$, but are they also sufficient for $\K=\Fq$? Indeed, if they are, then $k_{\Fq}\leq 2 k_{\C}$ follows immediately.  To see this, if there is a point $p$ with rank greater than $2k_{\C}$, then consider a line through $p$ and a point $q$ with rank $k_{\C}$.  This line has no other points with rank at most $k_{\C}$, since otherwise $p$ would be of rank no more than $2k_{\C}$, a contradiction. Therefore, the measure of the set of points with rank $k_{\C}$ must be less than a fraction $\frac{1}{q}$ of the whole space, which contradicts the assumption that $k_{\C}$ queries suffice.  Thus there is no point with rank greater than $2k_{\C}$---or in other words, if $k_{\C}$ queries are sufficient to achieve positive success probability, then $2k_{\C}$ queries are sufficient to achieve success probability $1$.

While we considered an algorithm with a bounded working region, it is unclear what is the highest success probability that can be achieved by a general $k$-query algorithm without this restriction (and in particular, whether fewer than $k_{\K}$ queries could suffice to solve the problem with high probability).  Indeed, even for the algorithm we proposed in \sec{algorithm}, it remains open to understand what choice of the region $S$ leads to the highest success probability.  As mentioned in \sec{opt}, it would be useful to establish lower bounds on the query complexity of polynomial interpolation over infinite fields.
Also, as stated in \cite{CvDHS16}, for the univariate case over finite fields, the algorithm is time efficient since the function $Z^{-1}(z)$, i.e., finding a preimage of elements in the range of $Z$, is efficiently computable.
However, for multivariate cases, it remains open whether there is an analogous efficiency analysis.

Finally, Zhandry has placed the quantum algorithm for polynomial interpolation in a broader framework that includes other problems such as polynomial evaluation and extrapolation \cite{Zha15}.  It could be interesting to consider these problems for multivariate polynomials and/or over infinite fields.

\section*{Acknowledgments}

We thank Charles Clark for encouraging us to consider quantum algorithms for polynomial interpolation over the real numbers. J.C. would also like to thank Jun Yu and Chi-Kwong Li for helpful comments in early discussions of the project.

This work received support from the Canadian Institute for Advanced Research, the Department of Defense, and the National Science Foundation (grant 1526380).


\end{document}